\documentclass[preprint,10pt]{elsarticle}

\usepackage{amssymb}
\usepackage{amsthm}
\usepackage[activeacute]{babel}
\usepackage{citesort}
\usepackage{graphicx}
\usepackage{amstext}
\usepackage{xcolor}
\usepackage{ulem}
\usepackage{url}
\usepackage{caption}
\usepackage{subcaption}

 \newtheorem{definition}{Definition}[section]
 
 \newtheorem{proposition}{Proposition}[section]

\usepackage{lineno}

\journal{Pyhsica A}

\begin{document}
\begin{frontmatter}



\title{Trend patterns statistics for assessing irreversibility in cryptocurrencies: time-asymmetry versus inefficiency}


\author[inst1]{Jessica Morales Herrera}
\affiliation[inst1]{
organization={Instituto de Investigación en Ciencias Básicas y Aplicadas, Universidad Autónoma del Estado de Morelos. },
            addressline={Avenida Universidad 1001}, 
            city={Cuernavaca},
            postcode={62209}, 
            state={Morelos},
            country={Mexico}}

\author[inst2]{R. Salgado-García}
\ead{raulsg@uaem.mx}
\affiliation[inst2]{organization={Centro de Investigación en Ciencias-IICBA, Universidad Autónoma del Estado de Morelos},
            addressline={Avenida Universidad 1001}, 
            city={Cuernavaca},
            postcode={62209}, 
            state={Morelos},
            country={Mexico}}

\begin{abstract}

In this paper, we present a measure of time irreversibility using trend pattern statistics. We define the irreversibility index as the Kullback-Leibler divergence between the distribution of uptrends subsequences (increasing trends) and the corresponding downtrends subsequences distribution (decreasing trends) in a time series. We use this index to analyze the degree of irreversibility in log return series over time, specifically focusing on five cryptocurrencies: Bitcoin, Ethereum, Ripple, Litecoin, and Bitcoin Cash. Our analysis reveals a strong indication of irreversibility in all these cryptocurrencies and the characteristic evolves over time. We additionally evaluate the market efficiency for these cryptocurrencies based on a recently proposed information-theoretic measure. By comparing inefficiency and irreversibility, we explore the relationship between these statistical features. This comparison provides insight into the non-trivial relationship between inefficiency and irreversibility.
\end{abstract}


\begin{highlights}
\item Trend patters statistics is used to asses irreversibility of time-series.
\item Time-irreversibility of some cryptocurrency markets is strongly varying in time.
\item Efficiency of cryptocurrencies is compared to time-irreversibility.
\end{highlights}

\begin{keyword}
keyword one \sep keyword two
\PACS 0000 \sep 1111
\MSC 0000 \sep 1111
\end{keyword}

\end{frontmatter}


\newpage
\section{Introduction}
\label{sec:introduction}

In 2009, Satoshi Nakamoto~\cite{nakamoto2008Bitcoin} launched the \textit{Bitcoin project}, which is a blockchain technology designed to facilitate peer-to-peer transactions and introduce digital currencies, now known as \textit{cryptocurrencies}. Since then, cryptocurrency markets have experienced rapid growth, attracting the attention of researchers, policymakers, investors, and regulators~\cite{corbet2019}. The latter is because cryptocurrencies possess several characteristics that make them appealing to investors. For example, some studies have suggested that Bitcoin has properties that make it a viable investment~\cite{dyhrberg2018investible}. Furthermore, the ability of Bitcoin, as well as other cryptocurrencies, to act as a hedge against stock markets has been analyzed~\cite{dyhrberg_hedging_2016,BOURI2017192}, suggesting that they can act as a safe haven~\cite{dyhrberg2018investible,singh2020}.

The increased interest in cryptocurrencies has prompted numerous studies. One of the primary characteristics examined in current literature is market efficiency. According to Fama~\cite{fama1970efficient}, the weak form of market efficiency asserts that the current asset price fully incorporates all available information. This is a key concept in financial market theory known as the \textit{Efficient Market Hypothesis} (EMH).
Assessing the level of inefficiency in a given market is of utmost importance since market inefficiency is closely related to price change predictability~\cite{eom2007relationship}. Due to its significance, market efficiency has been thoroughly examined by several authors. One of the earliest studies in this area was conducted by Urquhart, who performed various statistical tests and concluded that the Bitcoin market was inefficient but moving towards efficiency over time~\cite{urquhart2016inefficiency}.
Since then, a significant amount of research has been conducted on this topic~\cite{Bariviera2012Where}. To name a few examples, Cheah \textit{et al.} analyzed Bitcoin prices and found evidence of long-term memory, leading them to conclude that the Bitcoin market is inefficient in the weak form~\cite{CHEAH201818}. Bariviera~\cite{BARIVIERA20171} studied the dynamics of long-range dependence properties of Bitcoin prices and concluded that there is a trend towards efficiency in the Bitcoin market, as shown through a Hurst exponent analysis. Sensoy used permutation entropy to analyze the time-varying weak-form efficiency of Bitcoin prices in terms of US dollars (BTC-USD) and euro (BTC-EUR)~\cite{SENSOY201968}. Sensoy showed that both BTC-USD and BTC-EUR markets have become more informationally efficient over time. Other studies~\cite{Kyriazis2019Survey, Lopez-Martin2021Efficiency} support the same conclusion: the Bitcoin and other cryptocurrency markets are evolving over time in a way that leads to at least weak-form efficiency~\cite{kurihara2017market, tiwari2018informational, vidal2018semi, dyhrberg2018investible, vidal2019weak, Jena2022Evidence, ALYAHYAEE2020101168}.

The property of irreversibility, whether it is stochastic or deterministic, provides insight into the underlying dynamics, as noted by Hoover in his work~\cite{hoover2012time}. Physical systems that are out of thermodynamic equilibrium are associated with time-irreversibility. Non-linear processes in dynamical systems theory and non-Gaussian fluctuations are also linked to time-irreversibility, as found in research by Daw~\cite{Daw2000Symbolic}. Therefore, determining the degree of irreversibility of a process from a time series has been the subject of much research recently~\cite{Zanin2021Algorithmic}.
When analyzing financial time series, researchers have found that the concept of time irreversibility can be useful in understanding the dynamics of financial markets. This idea was first explored by Rothman~\cite{rothman1990characterization}, who studied time reversibility in financial time series to determine whether stock prices follow a random walk, which is a basic hypothesis in several economic models~\cite{rothman1994time}. Since then, many studies have shown that time irreversibility is an important concept for studying financial markets, both theoretically and empirically~\cite{Ramsey1996Time,Zumbach2009TimeReversal,irreversibility2016Flanagan}. For example, time irreversibility of financial time series has been used to analyze the dynamic relationship between shocks to volume and volatility~\cite{Fong2003Time}. Time irreversibility has also been applied to log return time series to classify financial markets~\cite{Xia2014Classifying}. As a result, various tests and methods have been developed to analyze the time irreversibility of financial time series~\cite{chen2000testing,irreversibility2016Flanagan,Li2018TimeIrreversibility,Wang2018NewMeasurement,zanin2018assessing}.

It is worth noting that some have suggested that irreversibility could be linked to market inefficiency. The argument is that time series that can't be reversed are less predictable than those that can, and since predictability and inefficiency are closely tied, it is evident that irreversibility could be used as a measure of inefficiency~\cite{irreversibility2016Flanagan}. As a result, it has been suggested to evaluate time series irreversibility to test the efficient market hypothesis~\cite{zanin2018assessing}.
\bigskip

The purpose of this work is twofold. Firstly, we introduce a new method to evaluate time-irreversibility through what we refer to as \emph{trend patterns}. Secondly, we analyze the relationship between irreversibility and inefficiency in some cryptocurrencies. We define a trend pattern as a sequence of log returns that either increase or decrease over time. We then measure the duration of increasing and decreasing patterns in a given time series. We argue that in a reversible time series, there would be no difference between the empirical distributions of the duration of increasing and decreasing patterns. Thus, measuring the statistical differences between such distributions (using the well-known Kullback-Leibler divergence) will give us a measure of the degree of irreversibility of the corresponding process. Additionally, we analyze the market inefficiency using a recently proposed approach based on information-theoretical indicators ~\cite{Brouty2022Statistical,Brouty2022Maxwell}.
Through our parallel analyses, we can determine the correlation between time-irreversibility and inefficiency in the market, allowing us to determine to what extent both indices might be equivalent.

As a result, this work is structured as follows: in Section~\ref{sec:irreversibility_trend}, we introduce the irreversibility index based on trend patterns and demonstrate its relationship to the entropy production rate in a specific Markov chain model. In Section~\ref{sec:estimation_procedure}, we outline the estimation procedure for assessing the trend irreversibility index and evaluate this method in both a reversible stochastic process (a second-order linear autoregressive process) and an irreversible one (a second-order nonlinear autoregressive process). In Section~\ref{sec:methodology}, we provide details about the methodology used for data pre-processing, as well as the statistical test with surrogate data. In this section, we also describe the methodology (introduced in Ref.~\cite{Brouty2022Statistical}) used to estimate inefficiency from log returns. Section~\ref{sec:results} is dedicated to presenting the results obtained from the statistical analyses conducted for each cryptocurrency data set. We conclude this work with a brief discussion of our results and a summary of our main findings in Section~\ref{sec:conclusions}.

\section{Irreversibility and trend pattern analysis}
\label{sec:irreversibility_trend}

\subsection{Irreversibility of a stochastic process}
\label{ssec:irreversibility}

\bigskip

Before we dive into the method for measuring irreversibility through the analysis of trend patterns, let us define some terms and notation that we will be using throughout this work. We will be focusing on discrete-time stochastic processes $\mathcal{X} := \{X_t \, : \, t\in \mathbb{N}_0\}$ with a finite state space $\mathcal{S}$. A finite random path or trajectory of the process starting at time $t$ and ending at time $t+n$ is denoted by $\mathbf{X}_t^{t+n}$, which we can represent as:
\begin{eqnarray}
\mathbf{X}_t^{t+n} := (X_t,X_{t+1},\dots,X_{t+n}),
\end{eqnarray}
this means that $\mathbf{X}_t^{t+n}$ is an $(n+1)$-dimensional random vector.

Bold lowercase letters such as $\mathbf{a}$, $\mathbf{b}$, etc. will denote realizations of the process or finite trajectories. We will use subscripts and superscripts to represent finite realizations $\mathbf{X}_t^{t+n}$ of the process $\mathcal{X}$, i.e.,
\begin{eqnarray}
\mathbf{a}_t^{t+n} := (a_t, a_{t+1}, \dots, a_{t+n}).
\end{eqnarray}

A stochastic process $\{X_t \, : \, t\in \mathbb{N}_0\}$ is said to be reversible if for all $t,n\in \mathbb{N}_0$ we have that
\begin{equation}
\mathbb{P}(\mathbf{X}_t^{t+n} = \mathbf{a}) = \mathbb{P}(\mathbf{X}_t^{t+n} = \mathbf{\overline{a}}),
\label{eq:def:reversibility}
\end{equation}
where $\mathbf{\overline{a}}$ stands for the time-reversed sample trajectory $\mathbf{a}$, i.e.,
\begin{eqnarray}
\mathbf{\overline{a}} :=  (a_{t+n},a_{t+n-1},\dots,a_{t}).
\end{eqnarray}
If the identity (\ref{eq:def:reversibility}) does not hold, the process is said to be \textit{irreversible}. One method to quantify the irreversibility of a given stochastic process is by using the Kullback-Leibler divergence between the distribution of direct and reversed trajectories~\cite{Chazottes2005Testing}:

\begin{eqnarray}
e_{\mathrm{p}} &:=& \lim_{n\to\infty} D_{\mathrm{KL}}\big(\,  \mathbb{P}(\mathbf{X}_0^{n} = \mathbf{a}) \,  || \, \mathbb{P}(\mathbf{X}_0^{n} = \mathbf{\overline{a}})\, \big)
\nonumber
\\
&=&
\lim_{n\to\infty}  \sum_{\mathbf{a}\in \mathcal{S}^{n+1}} \mathbb{P}(\mathbf{X}_0^{n}
= \mathbf{a}) \log\bigg( \frac{\mathbb{P}(\mathbf{X}_0^{n} = \mathbf{a}) }{ \mathbb{P}(\mathbf{X}_0^{n} = \mathbf{\overline{a}}) }\bigg).
\label{eq:def:ep}
\end{eqnarray}

This quantity, $e_{\mathrm{p}}$, is known as the entropy production rate in the context of Markov chains. It is also important to remark that $\log$ stands for the natural logarithm, a notation that we will use throughout all this work.

It is important to note that there are various ways to estimate the entropy production rate. For example, Ref.~\cite{Chazottes2005Testing} studied entropy production estimators based on recurrence time statistics. These estimators have been applied as irreversibility indices for multiple purposes~\cite{Yun2008Estimating, Salgado2021Estimating, Salgado2021Time}. However, the entropy production rate to measure irreversibility from a time series  requires a large sample size to achieve an acceptable estimation error. Therefore, determining the degree of irreversibility has been approached using a variety of methods~\cite{Zanin2021Algorithmic}. In this work, we propose an alternative method for determining irreversibility based on trend pattern statistics.

\subsection{Trend pattern distributions and irreversibility}
\label{ssec:irreversibility}

\bigskip

As mentioned in the introduction, we will assess the degree of irreversibility by analyzing trend patterns in the log return time series. Let $(\ell_{0},\ell_{1},\dots,\ell_{N})$ represent a log return sample data of size $N$. An \textit{uptrend} is defined as a subsequence $(\ell_{t},\ell_{t+1},\dots,\ell_{t+n})$ that occurs in the sample path and meets the following condition:
\begin{equation}
\ell_{t} < \ell_{t+1} < \ell_{t+2} < \cdots < \ell_{t+n},
\end{equation}
where $t\in \mathbb{N}_0$ and $n\in\mathbb{N}$. Similarly, a \textit{downtrend} is defined as a subsequence $(\ell_{0},\ell_{1},\dots,\ell_{N})$ that occurs in the sample path and meets the following condition,
\begin{equation}
\ell_{t} > \ell_{t+1} > \ell_{t+2} > \cdots > \ell_{t+n}.
\end{equation}
For completeness, we refer to those subsequences of a time series that are neither uptrends nor downtrends as \textit{constant trends}.

The length of time $T_{_\uparrow}(t)$ that a trend moves upwards starting at time $t$ is determined by,
\begin{equation}
T_{_\uparrow}(t) = \max\{ n \, : \, \ell_{t} > \ell_{t+1} > \ell_{t+2} > \cdots > \ell_{t+n}  \},
\end{equation}
Similarly, the length of time $T_{_\downarrow}(t)$ that a trend moves downwards starting at time $t$ is determined by,
\begin{equation}
T_{_\downarrow}(t) = \max\{ n \, : \, \ell_{t} < \ell_{t+1} < \ell_{t+2} < \cdots < \ell_{t+n}  \}.
\end{equation}
These random variables will be identified by their appropriate distributions, which are defined as follows:

\begin{definition}
Let $\mathcal{X} := \{X_t \, : \, t\in \mathbb{N}_0\}$ be a stochastic process and let 
$T_{_\uparrow}(0) $ be the duration of an uptrend starting at $t=0$. We 
denote by $P_{_{\uparrow}}(n)$ the probability that $T_{_\uparrow}(0)$ take the value $n\in \mathbb{N}$, i.e.,
\begin{equation}
P_{_{\uparrow}}(n):=\mathbb{P}(T_{_{{\uparrow}}}(0) = n).
\end{equation}

We denote by $P_{_{\downarrow}}(n)$ the probability that $T_{_\downarrow}(0)$ take the value $n\in \mathbb{N}$, i.e.,
\begin{equation}
P_{_{\downarrow}}(n):=\mathbb{P}(T_{_{{\downarrow}}}(0) = n).
\end{equation}
In the following we will refer to $P_{_{\uparrow}}$ and $P_{_{\downarrow}}$ as the \emph{uptrend distribution} and \emph{downtrend distribution} of the process $\mathcal{X}$, respectively.

\end{definition}

It is easy to observe that the uptrend and downtrend distributions coincide when the process $\mathcal{X}$ is a sequence of \textit{independent and identically distributed} (i.i.d.) random variables. This is due to the time symmetry of the process. If the probability of observing a sample path $\mathbf{a}\in \mathcal{S}^n$ equals the probability of observing the reversed path $\overline{\mathbf{a}}$, then the probability of observing an uptrend $\mathbf{u} \in \mathcal{S}^n$ is also equal to the probability of observing a downtrend $\mathbf{v} \in \mathcal{S}^n$. This is because the uptrend reversed in time, $\overline{\mathbf{u}}$, becomes a downtrend due to the change in time direction. Actuallly, this argument is valid for any reversible stochastic process.  Therefore, measuring the a difference between the uptrend and downtrend distributions, $P_{_{\uparrow}}$ and $P_{_{\downarrow}}$, would allows us to measure the degree of irreversibility. The Kullback-Leibler divergence is a well-known method for achieving this comparison.

\begin{definition}

Let $P_{_{\uparrow}}$ and $P_{_{\downarrow}}$ be the  uptrend and downtrend distributions of certain process $\mathcal{X}$. We define the \emph{trend irreversibility index} $I_\mathrm{T}$ as the Kullback-Leibler divergence between the  uptrend and downtrend distributions, i.e.,
\begin{equation}
I_{\mathrm{T}} := D_{\mathrm{KL}}\left( P_{_{\uparrow}} || P_{_{\downarrow}}  \right).
\label{eq:indexIT}
\end{equation}

\end{definition}

We will now establish a result that enables us to relate the trend of irreversibility index to the entropy production rate in a random walk on a lattice.

\begin{proposition}
\label{prop:DKLup-down}
Consider a random walk $\mathcal{X}$ on $\mathbb{Z}$, where $\mathcal{X} := \{X_t \, : \, t\in \mathbb{N}_0\}$. Let $P_{_{\uparrow}}$ and $P_{_{\downarrow}}$ represent the uptrend and downtrend distributions of $\mathcal{X}$, respectively. Assuming the walker moves to the right with probability $p$ in one time step, moves to the left with probability $1-p$, and starts its motion at the origin, we have:
\begin{equation}
D_{\mathrm{KL}}\left( P_{_{\uparrow}} || P_{_{\downarrow}}  \right) = \frac{e_{\mathrm{p}}}{1-p}.
\end{equation}
\end{proposition}

\begin{proof}

A random walk on $\mathbb{Z}$ is a Markov chain for which the probability of a path is given by 
\begin{equation}
\mathbb{P}(\mathbf{X}_0^n = \mathbf{a}) = \pi_{0}(a_0) M(a_0,a_1) M(a_1,a_2) \dots M(a_{n-1},a_{n}),
\end{equation}
where $\pi_0 : \mathbb{Z} \to [0,1]\subset \mathbb{R}$ is the initial probability vector (the initial distribution) and $M : \mathbb{Z}^2 \to [0,1]\subset \mathbb{R} $ is the corresponding stochastic matrix given by
\begin{equation}
M(a,b) = p \delta_{a+1,b} + (1-p)\delta_{a-1,b}, \quad \forall a,b\in \mathbb{Z}.
\end{equation}
In the above equation the symbol $\delta_{i,j}$ stands for the Kronecker delta, which is defined as
\[
\delta_{i,j} :=  \left\{ \begin{array} 
            {l@{ \quad \mbox{ if } \quad }l} 
1  & i=j,\\
0 & i\not= j,
             \end{array} \right.
\]

We can observe that a sample trajectory of the random walk only consists of either up or down trends. Furthermore, there is a single uptrend $\mathbf{u}$, which starts at $X_0 = 0$ and lasts for $n$ time units, i.e., $\mathbf{u} = (0,1,2,\dots,n)$. There is also only one downtrend that starts at $X_0 = 0$, lasts for $n$ units, and is given by $\mathbf{v} = (0,-1,-2,\dots,-n)$. Therefore, we can easily compute the corresponding trend distributions. For the uptrend distribution, we have that,
\begin{eqnarray}
 P_{_{\uparrow}}(n)  &=& \mathbb{P}(\mathbf{X}_0^{n} = (0,1,2,\dots,n-1,n); X_{n+1} = n-1)
 \nonumber
 \\
 &=& p^{n} (1-p).
\end{eqnarray}
Analogously,  for the downtrend distribution we obtain,
\begin{eqnarray}
 P_{_{\downarrow}}(n)  &=& \mathbb{P}(\mathbf{X}_0^{n} = (0,-1,-2,\dots,-n+1,-n); X_{n+1} = -n+1)
 \nonumber
 \\
 &=&(1- p)^{n} p.
\end{eqnarray}
These results show that the up and down trend duration for the random walk are distributed according to a geometric distribution. Thus, the Kullback-Leibler divergence  $D_{\mathrm{KL}}\left( P_{_{\uparrow}}(n) || P_{_{\downarrow}}(n)  \right) $ can be exactly computed, 
\begin{eqnarray}
D_{\mathrm{KL}}\left( P_{_{\uparrow}} || P_{_{\downarrow}}  \right) & = &
\sum_{n=0}^\infty P_{_{\uparrow}}(n) \log\bigg(\frac{P_{_{\uparrow}}(n)}{P_{_{\downarrow}}(n)} \bigg) 
\nonumber
\\
& = &
\sum_{n=0}^\infty  p^{n} (1-p)  \log\bigg(\frac{ p^{n} (1-p) }{p (1-p)^n} \bigg) 
\nonumber
\\
& = &
(1-p) \sum_{n=0}^\infty n p^{n}   \log\bigg(\frac{ p  }{ 1-p} \bigg)
+ (1-p) \sum_{n=0}^\infty  p^{n}   \log\bigg(\frac{ 1-p  }{ p} \bigg).
\nonumber 
\end{eqnarray}
It is not hard to see that the summations can be straightforwardly evaluated, obtaining, 
\begin{eqnarray}
D_{\mathrm{KL}}\left( P_{_{\uparrow}} || P_{_{\downarrow}}  \right) 
&=& 
\frac{p}{1-p}  \log\bigg(\frac{ p  }{ 1-p} \bigg)
+ \log\bigg(\frac{ 1-p  }{ p} \bigg)
\nonumber
\\
&=& 
\frac{2p-1}{1-p}  \log\bigg(\frac{ p  }{ 1-p} \bigg).
\label{eq:DKL-up|down-1}
\end{eqnarray}

On the other  hand, it is known that the entropy production rate for a random walk on a $\mathbb{Z}$ can be written as~\cite{Cocconi2020Entropy},
\begin{eqnarray}
e_{\mathrm{p}} = (2p-1) \log\bigg(\frac{ p  }{ 1-p} \bigg).
\end{eqnarray}
This expression and Eq.~(\ref{eq:DKL-up|down-1}) allows us to see that,
\begin{eqnarray}
D_{\mathrm{KL}}\left( P_{_{\uparrow}} || P_{_{\downarrow}}  \right) 
&=& \frac{2p-1}{1-p}  \log\bigg(\frac{ p  }{ 1-p} \bigg) 
 = \frac{e_{\mathrm{p}} }{1-p},
\label{eq:DKL-up|down}
\end{eqnarray}
which proves the proposition.

\end{proof}

It is worth noting that $D_{\mathrm{KL}}\left(  P_{_{\downarrow}} ||   P_{_{\uparrow}} \right)$ for the random walk can be calculated using Proposition~\ref{prop:DKLup-down}.
For this stochastic process we can easily obtain $P_{_{\downarrow}}(n)$ by substituting $1-p$ for $p$ in the expression for $P_{_{\uparrow}}(n)$, and vice versa. This property allows us to write down the following identity directly from Eq.~(\ref{eq:DKL-up|down}):
\begin{eqnarray}
D_{\mathrm{KL}}\left(  P_{_{\downarrow}} ||   P_{_{\uparrow}}  \right) 
&=& \frac{1-2p}{p}  \log\bigg(\frac{ 1-p  }{ p} \bigg) = \frac{e_{\mathrm{p}}}{p},
\label{eq:DKL-down|up}
\end{eqnarray}
which relates the Kullback-Leibler divergence $D_{\mathrm{KL}}\left(  P_{_{\downarrow}} ||   P_{_{\uparrow}} \right) $ to the entropy production rate $e_p$.

\section{Estimation of trend irreversibility index}
\label{sec:estimation_procedure}

\subsection{Estimation procedure}
\label{ssec:estimation}

In this section, we will establish a method for estimating both the uptrend and downtrend distributions, which will be used to estimate $I_{\mathrm{T}}$, giving us the degree of time irreversibility of the process.

Firstly, we should note that the time series represents a stationary process. Next, we can see that a time series can be divided into uptrends, downtrends, and constant trends, as shown in Fig.~\ref{fig:fig01}. Therefore, the duration of each uptrend and downtrend can be considered as a realization of $T_{_{{\uparrow}}}(0)$ and $T{_{_{\downarrow}}}(0)$, respectively. Thus, by collecting the durations of these trend patterns, we can estimate $P_{_{\uparrow}}$ and $P_{_{\downarrow}}$.

%
\begin{figure}[ht]
\begin{center}
\scalebox{0.5}{\includegraphics{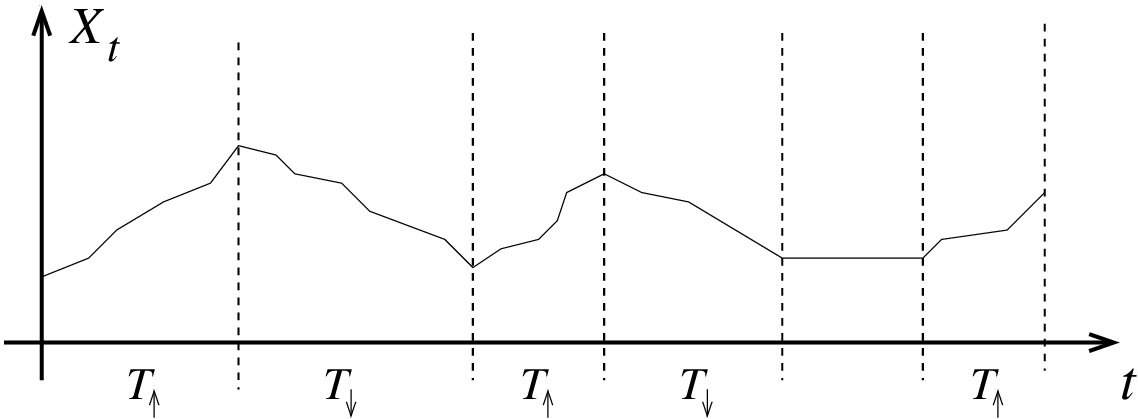}}
\end{center}
     \caption{
			This is a schematic representation of trend patterns in a time series. A time series $\{x_t : t\in \mathbb{N}0\}$ can be divided into uptrends, downtrends, and periods of constancy. The duration of uptrends is labeled as $T{{{\uparrow}}}$, while the duration of downtrends is labeled as $T{_{{\downarrow}}}$. This results in sample sets for periods of uptrends, downtrends, and periods of constancy.
			}
\label{fig:fig01}
\end{figure}
%

To exemplify the estimation procedure, we consider a simulation of the random walk. Figure~\ref{fig:fig02} shows the numerically estimated trend distributions $\hat{P}_{_{{\uparrow}}}$ and $\hat{P}_{_{{\downarrow}}}$ for different values of $p$. We used $p = 0.6$ and $N=10^5$ time-steps to obtain these estimations. In Figure~\ref{fig:fig02}a, we see the unbiased case, where $p=0.5$. $\hat{P}_{_{{\uparrow}}}$ and $\hat{P}_{_{{\downarrow}}}$ are close to each other, as expected from the time-reversible trajectories. In Figures~\ref{fig:fig02}b, \ref{fig:fig02}c, and \ref{fig:fig02}d, we display the uptrend and downtrend distributions for $p=0.60$, $p=0.70$, and $p=0.90$. As $p$ increases, the trend distributions $\hat{P}_{_{{\uparrow}}}$ and $\hat{P}_{_{{\downarrow}}}$ move away from each other, indicating time irreversibility, as we proved in Proposition~\ref{prop:DKLup-down}.

%
\begin{figure}[ht]
\begin{center}
\scalebox{0.5}{\includegraphics{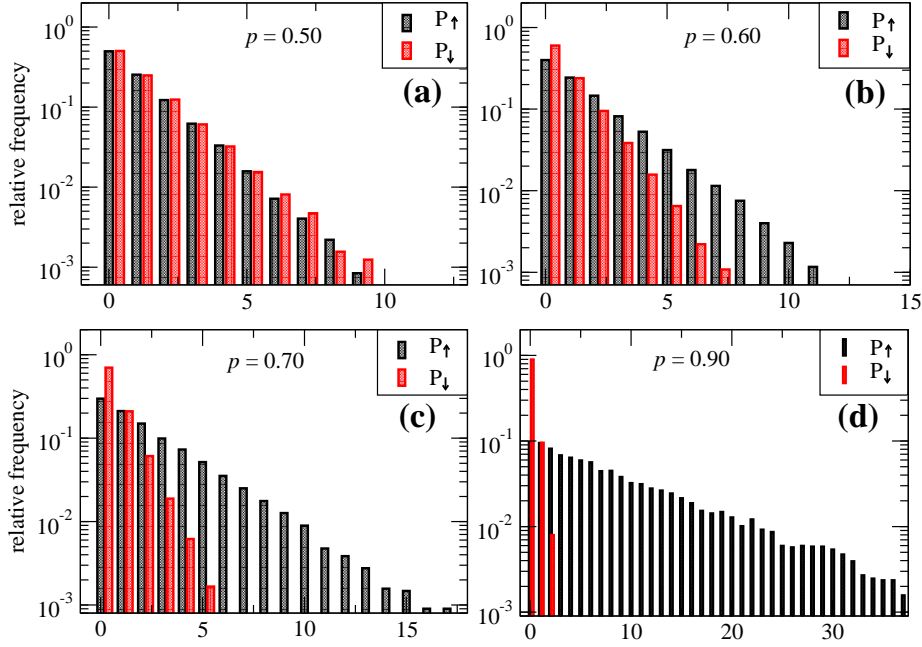}}
\end{center}
     \caption{
			Estimation uptrend and  downtrend distributions for a random walk on $\mathbb{Z}$.
	     }
\label{fig:fig02}
\end{figure}
%

\subsection{Testing the trend irreversibility index}
\label{ssec:estimation}

In this section we test the trend irreversibility index by using time-series obtained by numerical simulations of two processes: (i) a (second order) linear autoregressive (AR) process  which is known to be reversible and (ii) a non-linear autoregressive process which is known to be time irreversible~\cite{martinez2018Detection}. 

The second-order in time linear autoregressive process (AR) we consider here is the one studied in Ref.~\cite{martinez2018Detection}, and is given by the following recurrence equation,
\begin{equation}
\label{eq:AR2}
x_{t+2} = 0.7x_{t+1} + 0.2x_t + \xi_t,
\end{equation}
where $\xi_t$ is a gaussian white  noises, i.e., $\{\xi_t : t\in \mathbb{N}_0\}$ is a set of i.i.d. normally distributed random variables with zero mean and variance one.

On the other hand, the second order in time non-linear autoregressive (NAR) process we consider is defined as follows:
\begin{eqnarray}
\nonumber
x_{t+2} &=& 0.5x_{t+1}- 0.3x_{t} + 0.1y_{t} + 0.1x^2_{t} 
\\
\label{eq:NAR2-1}
&+& 0.4y^2_{t+1} + 0.0025\eta_t,
\\
\label{eq:NAR2-2}
y_{t+2} &=&\sin(4\pi t)+\sin(6\pi t)+0.0025\zeta_t.
\end{eqnarray}
The Laplacian noises in the equation above refer to two sets of independent and identically distributed random variables, namely $\{\eta_t : t\in \mathbb{N}_0\}$ and $\{\zeta_t : t\in \mathbb{N}_0\}$, with a Laplace distribution $\mbox{Lap}(\mu,\beta)$ as explained in~\cite{martinez2018Detection}. We set the parameters of the distribution to $\mu = 0$ and $\beta = 1$. The Laplace distribution $\mbox{Lap}(\mu,\beta)$ is defined by its probability density function,
\begin{equation}
f_\mathrm{Lap}(x) = \frac{1}{2\beta}e^{|x-\mu|/\beta}.
\end{equation}

To test the irreversibility index based on trend patterns, we simulate both the AR and NAR processes numerically. In the case of the NAR process, which is bivariate, we obtain a single time series by defining $u_t = x_t^2 + y_t^2$. We use the synthetic time series obtained from these simulations to estimate the uptrend and downtrend distributions. These distributions then allow us to estimate the irreversibility index $I_{\mathrm{T}}$ according to Equation (\ref{eq:indexIT}).

\begin{figure}[t]
\begin{center}
\scalebox{0.4}{\includegraphics{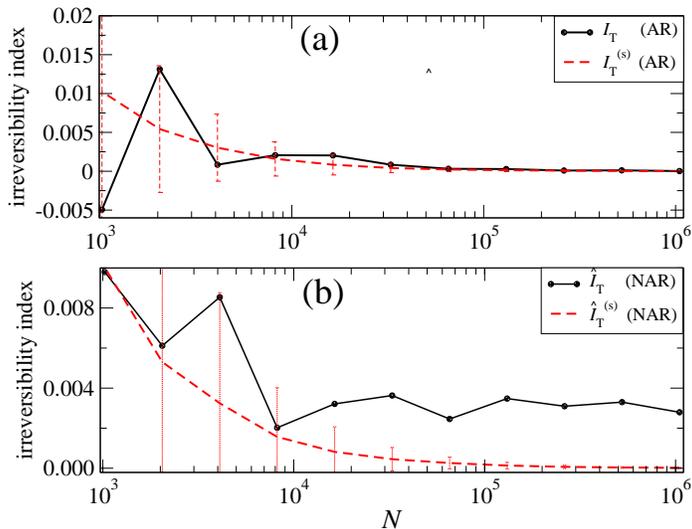}}
\end{center}
     \caption{
      The trend irreversibility index $I_{\mathrm{T}}$ is shown for autoregressive processes varying the sample size. Panel (a) displays the results for the reversible AR process, with the irreversibility index $I_{\mathrm{T}}$ (solid black line) being statistically negligible, as expected, when compared to the threshold of significance (dashed red line with error bars). Panel (b) shows the results for the irreversible NAR process, with the index $I_{\mathrm{T}}$ able to detect the irreversibility for sample sizes above $10^4$.
            }
\label{fig:fig03}
\end{figure}

To evaluate the performance of the irreversibility index as a function of the time-series length $N$, we estimated $\hat I_{\mathrm{T}}$ for several values of $N$. In addition, we used an ensemble of $N_{s}=100$ surrogate time series to evaluate the significance of $\hat I_{\mathrm{T}}$ by calculating the irreversibility index $I_{\mathrm{T}}$ again; this estimation is referred to as $\hat I_{\mathrm{T}}^{(s)}$. The results are shown in Fig.~\ref{fig:fig03}.
Fig.~\ref{fig:fig03}a shows the behavior of $\hat I_{\mathrm{T}}$ as a function of the time series length $N$ (black solid line) for the AR process described above. The red solid line depicts the dependence of $I_{\mathrm{T}}^{(s)}$ with the time series length. The error bars denote confidence intervals with a confidence level of $1-\alpha$ for $\alpha = 0.05$. The irreversibility index $I_{\mathrm{T}}$ is able to determine the reversible nature of the process for all values of the time-series length analyzed.
Fig.~\ref{fig:fig03}b shows the behavior of $\hat I_{\mathrm{T}}$ as a function of the time series length $N$ (black solid line) for the NAR process described above. The red solid line depicts the dependence of $I_{\mathrm{T}}^{(s)}$ with the time series length, and the error bars correspond to confidence intervals with a confidence level of $1-\alpha$ for $\alpha = 0.05$. For time-series length larger than $N^*\approx 10^4$ time steps, the irreversibility index $I_{\mathrm{T}}$ can correctly determine the irreversibility of the process. For values of $N$ below $N^*$, the difference between $\hat I_{\mathrm{T}}$ and $I_{\mathrm{T}}^{(s)}$ is not significant. However, for time-series lengths larger than $N^*$, a significant difference is observed, which allows us to conclude that the process is irreversible, as expected.
This analysis indicates that our method for determining the irreversibility from a time-series is at least as efficient as other methods available in the current literature, with a significantly low computational cost.

\section{Methodology}
\label{sec:methodology}

In this analysis, we have examined the primary cryptocurrencies, including Bitcoin (BTC), Ethereum (ETH), and Ripple (XRP), which have become increasingly popular in recent years. Bitcoin is renowned for being the first-ever cryptocurrency with a maximum supply of 21 million coins, which makes it a deflationary currency. It is currently the largest cryptocurrency by market capitalization, which exceeded 500 billion USD in April 2023~\cite{CoinMark}. On the other hand, Ethereum does not have a fixed number of coins and currently holds a market cap of over 200 billion USD, placing it in second place. XRP holds the sixth place in market capitalization with over 26 billion dollars~\cite{CoinMarket}. 

The analysis also included Litecoin (LTC) and Bitcoin Cash (BCH), two currencies that, while not as prominent as others, are still significant due to their market capitalization. Litecoin, which was created by Charlie Lee in 2011, is often considered the ``silver to Bitcoin's gold''~\cite{Mustafa}. It is designed to be faster and less expensive than Bitcoin. Bitcoin Cash, on the other hand, was created as a fork of the original Bitcoin blockchain in 2017. Some proponents argue that it is an improved version of the original cryptocurrency and that its goal is to stay true to the vision of the creator, Satoshi Nakamoto~\cite{bch}


We acquired data from the CryptoDataDownload website \citep{CrypDownload}, which offers historical databases from various exchanges such as Bitstamp, Bitso, Gemini, Binance, and others. For each exchange, we had nine columns that included Unix time, date, symbol, open, high, low, close, volume of cryptocurrency, and volume.
Using the aforementioned cryptocurrency data, we calculated the logarithmic returns of the opening prices between times $t$ and $t-1$, defined as follows.

\begin{definition}
The log return $\ell_t$ are defined as the natural logarithm of the ratio between the price at time $t$ and the price at time $t-1$, i.e., 
\begin{equation}
\ell_t = \log \left( \frac{P_{t}}{P_{t-1}} \right),
\end{equation}
where $P_t$ is the price at time $t$, and $P_{t-1}$ is the price at time $t-1$.
\end{definition}

\subsection{Data Preprocessing}
\label{ssec:data-preprocessing}

We were able to access several years of data for each cryptocurrency at a minute frequency, although the available data varied. For Bitcoin, our data spanned from 2015 to 2023, but we only analyzed up to 2022 to ensure we included full years. However, we did encounter some missing values in the retrieved data. We summarized the fraction of missing values we found for each data set used in our analysis in Table~\ref{tab:criptable}.

To address this issue, we replaced the missing data in the log returns with independent random numbers following a normal distribution. We used the sample mean and sample standard deviation of the corresponding log returns as the mean and standard deviation of this normal distribution, respectively. This substitution was necessary to avoid a cumulative time delay in the log return time series. It is important to note that the substitution of missing data by independent and identically distributed random numbers does not contribute to the irreversibility nor inefficiency of the corresponding cryptocurrency data (see Section~\ref{ssec:efficiency} below).

\begin{table}[h]
\begin{center}
\begin{tabular}{| c | c | c |}
\hline
Cryptocurrency & Analysis Period & Amount of missing data \\ \hline
BTC & 2015-2022 & $ 0.9802\, \% $ \\
ETH & 2017-2022 & $ 3.4315\, \% $\\
XRP & 2017-2022 & $ 0.0582\, \% $ \\
BCH & 2018-2022 & $ 0.0002\, \% $ \\
LTC & 2019-2022 & $ 3.5554\, \% $ \\ \hline
\end{tabular}
\caption{Cryptocurrency information. In the table, we display the amount of missing data and the period of analysis for each cryptocurrency used in this study.
}
\label{tab:criptable}
\end{center}
\end{table}

\subsection{Analysis of Trends Patterns}
\label{ssec:analisis-trend}

In order to track how irreversibility changes over time in cryptocurrency datasets, we utilize the sliding windows technique. We have set the window size to 91 days (or $131\, 040$ minutes), and estimate irreversibility and inefficiency based on the data within that window. After we obtained these estimations, we move the window ten days forward ($14\, 440$ minutes) and repeat the estimations. We chose this window size to ensure that the sample size is large enough to minimize statistical errors, as our tests in Section~\ref{ssec:estimation} showed. It is important to note that, due to the continuous nature of the datasets, some analyzed quarters may include data from two consecutive years.

Next, we calculated the durations $T_{_\uparrow}$ and $T_{_\downarrow}$ for uptrends and downtrends, respectively, for the log returns in each sliding window. We then compared the distributions $P_{_\uparrow}$ and $P_{_\downarrow}$ of uptrends and downtrends for both datasets by computing the Kullback-Leibler divergence $I_\mathrm{T} = D_{\mathrm{KL}}(P_{_\uparrow} ||P_{_\downarrow} )$. This approach allowed us to quantify the difference between these distributions.

As previously mentioned, if the process is reversible, the irreversibility index $I_{\mathrm{T}}$ is expected to be statistically insignificant. Otherwise, the system can be considered to exhibit some degree of irreversibility. To determine the statistical significance, we performed a surrogate data analysis by randomly shuffling the original data (log returns) to eliminate any sign of irreversibility in the sample. An ensemble of surrogate data was used as a null hypothesis (zero irreversibility) to examine the significance of the irreversibility obtained through $I_\mathrm{T}$.

\subsection{Efficiency analysis}
\label{ssec:efficiency}

In this section, we will be using the methodology proposed by Brouty and Garcin~\cite{Brouty2022Statistical} to assess inefficiency (in the weak sense) from the cryptocurrency prices time series. To begin, we consider a set of $n+1$ consecutive log returns, denoted by $\{\ell_0,\ell_1,\dots,\ell_n\}$. This sample set will then be transformed into a new time series $\{ x_i\, :\, 1\leq i \leq n \}$, with a binary state space as outlined below,
\begin{equation}
a_i =  \left\{ \begin{array} 
            {l@{ \quad \mbox{ if } \quad }l} 
1  & \ell_i >0,\\
0 &  \ell_i \leq 0,
             \end{array} \right.
\end{equation}
Assuming that $\{ a_i\, :\, 1\leq i \leq n \}$ represents a finite realization of a stochastic process denoted by $\{X_{t}\, : \, t\in \mathbb{N}_0\}$, the probability of obtaining a given trajectory $\mathbf{a} = (a_1,a_1,\dots,a_L)$ of length $L$ and binary state space ($a_i \in \{0,1\}$) can be expressed in the following way, 
\begin{equation}
p(\mathbf{a}) := \mathbb{P}\left((X_1,X_2,\dots,X_{L})=\mathbf{a} \right).
\end{equation}
Using this notation, we can define the joint entropy for the discrete random vector $(X_1,X_2,\dots,X_{L})$ as stated in~\cite{cover1999elements},
\begin{equation}
H^L :=  H(X_1,X_2,\dots,X_{L}) = - \sum_{\mathbf{a} }p(\mathbf{a})\log\left(p(\mathbf{a}) \right).
\label{eq:HL} 
\end{equation}

If a market is not working efficiently, statistical arbitrages can appear, and it may be possible to take advantage of forecasting the value of $X_{L+1}$ by analyzing the past history of the process $(X_1,X_2,\dots,X_L)$. Brouty and Garcin argue that if the corresponding market satisfies the EMH (in the weak sense), we could not take advantage of the market and therefore, we would be unable to predict $X_{L+1}$ using past information of the process. This implies that the random variable $X_{L+1}$ can be considered as a fair Bernoulli trial. Conversely, if the market allows the possibility of arbitrage, we would be able to forecast the future value of the process, i.e., of $X_{L+1}$, and then take advantage of such a prediction. Brouty and Garcin proposed in Ref.~\cite{Brouty2022Statistical} to distinguish between efficient and inefficient markets by evaluating the difference in joint entropies of $(X_1,\dots,X_L, X_{L+1})$ under these two scenarios. They introduced the inefficiency index, defined as,
\begin{equation}
I_*^{L+1}= H_*^{L+1}-H^{L+1},
\end{equation}
where $H_*^{L+1}$ represents the joint entropy of $(X_1,\dots,X_L, X_{L+1})$, assuming that $X_{L+1}$ is an independent fair Bernoulli trial. The variable $H^{L+1}$ is the joint entropy for $(X_1,\dots,X_L, X_{L+1})$, as defined in Equation~(\ref{eq:HL}). Brouty and Garcin showed that $H_*^{L+1}$ can in fact be expressed as follows,
\begin{equation}
H_*^{L+1} = 1 + H^L,
\end{equation}
this enables us to express the inefficiency index $I_*^L$ in a straightforward manner as,
\begin{equation}
I_*^{L+1}= \log(2) + H^L -H^{L+1}.
\end{equation}
As stated in Theorem 1 from Ref.~\cite{Brouty2022Statistical}, the market will only be efficient if $I_*^{L+1}=0$, and inefficient otherwise. It is worth noting that we determine the statistical significance of the inefficiency index $I_*^{L+1}$ by performing a surrogate data analysis, following the same method as the irreversibility index (refer to Section~\ref{ssec:analisis-trend}).

\section{Results}
\label{sec:results}

%
\begin{figure}[ht]
\begin{center}
\scalebox{0.5}{\includegraphics{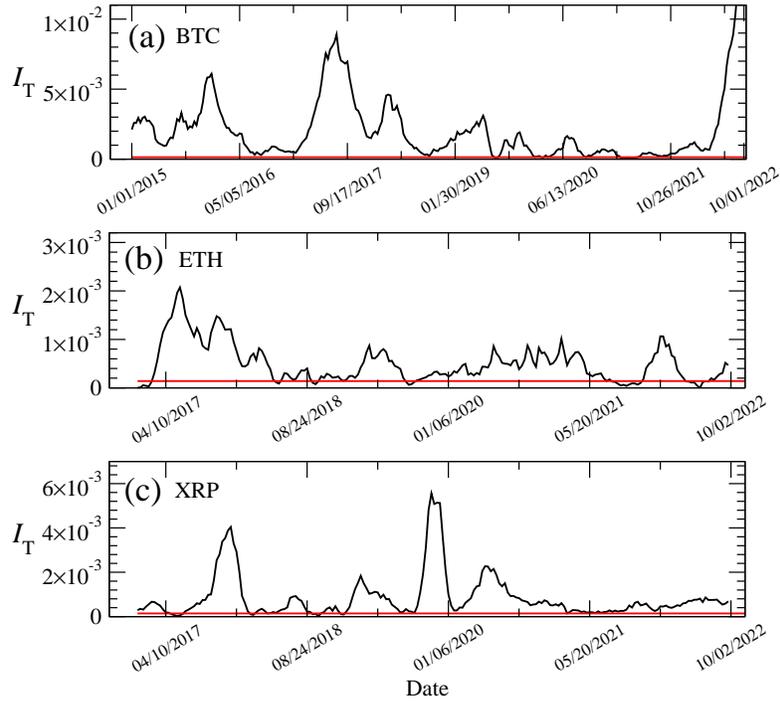}}
\end{center}
     \caption{
The trend irreversibility index ($I_{\mathrm{T}}$) was estimated from the log return time series of Bitcoin (BTC), Ethereum (ETH), and XRP, as shown in (a), (b), and (c) respectively. The irreversibility indices are represented by solid black lines, while the solid red lines indicate the threshold of statistical confidence at $95\%$. 
      }
\label{fig:IT:BTC-ETH-XRP}
\end{figure}
%

%
\begin{figure}[ht]
\begin{center}
\scalebox{0.5}{\includegraphics{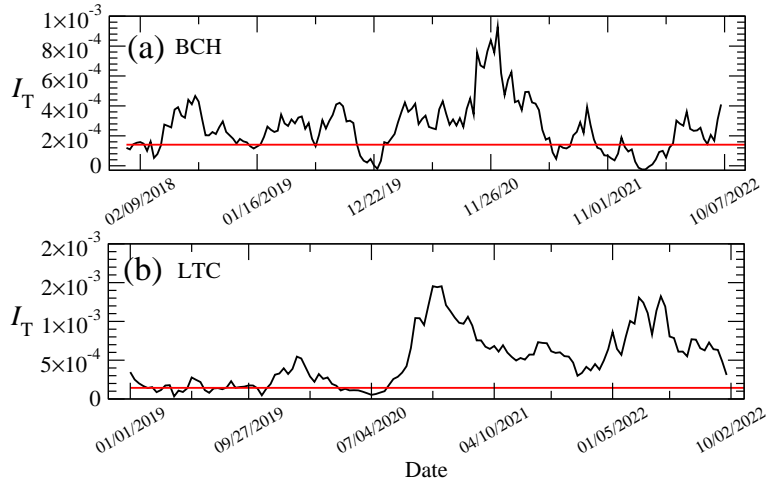}}
\end{center}
     \caption{
The trend irreversibility index $I_{\mathrm{T}}$ was estimated from the log return time series of both Bitcoin Cash (BCH) and Litecoin (LTC). The solid black lines represent the irreversibility indices, while the solid red lines indicate the threshold of statistical confidence at $95\%$.       }
\label{fig:IT:BCH-LTC}
\end{figure}
%

We start by showing, in Fig.~\ref{fig:IT:BTC-ETH-XRP},  the estimated trend irreversibility index $I_{\mathrm{T}}$ we obtained from (a) Bitcoin (BTC), (b) Ethereum (ETH) and (c) XRP log return time series. 
From these irreversibility curves we  can observe that these cryptocurrency markets have strong signs of irreversibility, something which is clear when comparing the estimated curves (solid black lines) with the threshold of statistical significance (solide red lines), which as we mentioned above was obtained by means of surrogate data. Above such a threshold, we say that the irreversibility is positive with a confidence at $95\%$. We should also notice that this irreversibility exhibited in all these cryptocurrencies is not constant: we clearly see periods of high irreversibility followed from periods of low (or even statistically vanishing) irreversibility. Also, but less clear at first sight, we could even say that the irreversibility seems to decrease globally overtime although not in a monotonic nor a systematic way.

In Fig.~\ref{fig:IT:BCH-LTC}, we show the trend irreversibility index $I_{\mathrm{T}}$ estimated from the log returns of Bitcoin Cash (BCH) and Litecoin (LTC). We observe that the irreversibility curves for these two cryptocurrencies exhibit a similar pattern to the first three cryptocurrencies (BTC, ETH, and XRP), with periods of high and low irreversibility. However, in contrast to the other cryptocurrencies, we did not observe a systematic decrease in the irreversibility over time for BCH and LTC. Bitcoin Cash demonstrates a consistently low level of irreversibility, as indicated by the significance threshold. For Litecoin, we see an increasing trend in irreversibility over the analyzed period. It is important to note that Litecoin was launched in 2011 and the analyzed period is relatively short. Therefore, the observed increase in irreversibility may not be representative of the cryptocurrency overall behavior. Nevertheless, it is evident that all analyzed cryptocurrencies display signs of varying irreversibility, a characteristic that reflects the non-stationarity of cryptocurrency markets.
%
%
\begin{figure}[ht]
\begin{center}
\scalebox{0.5}{\includegraphics{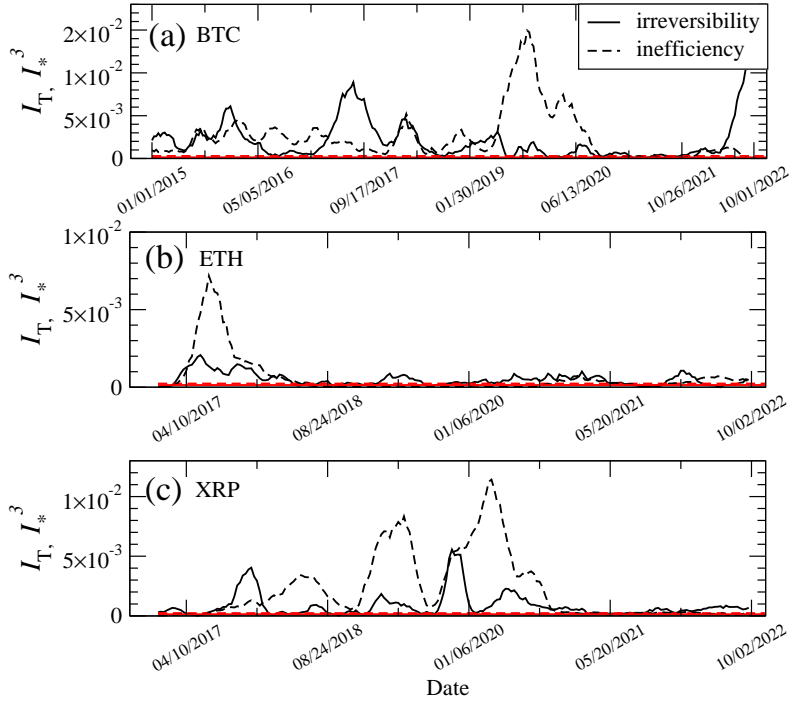}}
\end{center}
     \caption{
Trend irreversibility index $I_{\mathrm{T}}$ and inefficiency index $I_*^{3}$ estimated from (a) Bitcoin (BTC), (b) Ethereum (ETH) and (c)  XRP log return time series. Solid black lines represent the irreversibility and the dashed black lines represent the inefficiency.   Solid and dashed red lines stand for the threshold of statistical confidence at $95\%$ for irreversibility and inefficiency respectively.
      }
\label{fig:IT:BTC-ETH-XRP-compara}
\end{figure}
%

%
\begin{figure}[ht]
\begin{center}
\scalebox{0.5}{\includegraphics{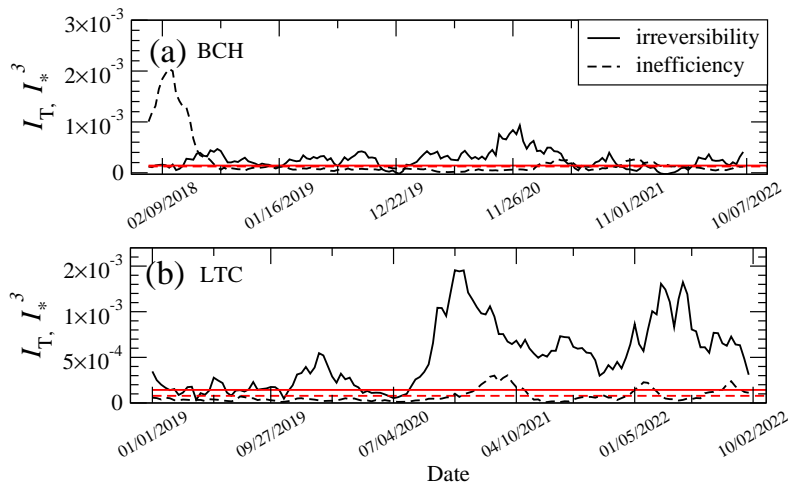}}
\end{center}
     \caption{
Trend irreversibility index $I_{\mathrm{T}}$ and inefficiency index $I_*^{3}$  estimated from (a) Bitcoin Cash (BCH), (b) Litecoin (LTC) log return time series.  Solid black lines represent the irreversibility and the dashed black lines represent the inefficiency.  Solid and dashed red lines stand for the threshold of statistical confidence at $95\%$ for irreversibility and inefficiency respectively.
       }
\label{fig:IT:BCH-LTC-compara}
\end{figure}
%

To analyze the relationship between inefficiency and irreversibility, we evaluated the inefficiency index $I^{3}_*$ for all the cryptocurrencies mentioned. Fig.~\ref{fig:IT:BTC-ETH-XRP-compara} shows a comparison of the irreversibility index (solid black lines) and the inefficiency index (dashed black lines) for (a) BTC, (b) ETH, and (c) XRP cryptocurrencies. The comparison reveals that inefficiency bears some similarity to irreversibility, with periods of high inefficiency alternating with periods of low or null inefficiency. These are what we can see as ``peaks'' of inefficiency, similar to the peaks of irreversibility we found in the time evolution of such indices. Fig.~\ref{fig:IT:BTC-ETH-XRP-compara} also shows that besides the characteristics shared by irreversibility and inefficiency, there are some periods in which both indices seem to coincide, but in others, they do not. For instance, irreversibility and inefficiency indices estimated from ETH log returns behave almost in the same way throughout the analyzed period. However, for BTC and XRP, we only observe small time intervals in which irreversibility and inefficiency follow each other. It is worth noting that at first sight, it is not completely clear that irreversibility and inefficiency are equivalent, at least within the extent of this statistical analysis.

Something similar occurs with the other two cryptocurrencies, BCH and LTC. In Fig.~\ref{fig:IT:BCH-LTC-compara}, we can see a comparison between the estimated irreversibility and inefficiency from BCH and LTC log returns.

In these cases, the difference between irreversibility and inefficiency seems to be more noticeable. If we look at Fig.~\ref{fig:IT:BCH-LTC-compara} (a), which shows the time evolution of irreversibility and inefficiency for BCH, we can see that the inefficiency (dashed black line) is significant at the beginning of the analyzed period, but quickly reaches efficiency. In contrast, the time evolution of irreversibility for BCH exhibits periods with positive irreversibility throughout the analyzed period. For Litecoin, the inefficiency is almost zero at the level of statistical significance we use, as shown in Fig.~\ref{fig:IT:BCH-LTC-compara}(b). However, for Litecoin, we also see that the irreversibility is positive and particularly large at the end of the analyzed period. To evaluate how much the inefficiency and irreversibility are related to each other, we computed the Pearson correlation coefficient between the irreversibility and inefficiency indices, as shown in Table~\ref{tab:cripcorr}. As expected, the correlation between irreversibility and inefficiency is small for BTC and LTC, while for the other cryptocurrencies, such correlation is relatively large. However, these correlation values are not large enough to suggest that inefficiency and irreversibility are completely equivalent, at least within the extent of our statistical analysis.
\begin{table}[h]
\begin{center}
\begin{tabular}{| c | c |}
\hline
Cryptocurrency & Correlation coefficient
\\ 
\hline
BTC & $-0.1111$  
\\
ETH & $0.6733$ 
\\
XRP & $0.3253$ 
\\
BCH & $0.4611$ 
\\
LTC & $-0.2087$ 
\\ \hline
\end{tabular}
\caption{Pearson correlation coefficient between $I_{\mathrm{T}}$ and $I^3_{*}$ for all the analyzed cryptocurrencies.}
\label{tab:cripcorr}
\end{center}
\end{table}
%
%

\section{Discussion  and conclusions}
\label{sec:conclusions}

In this study, we analyzed the evolution of time-irreversibility in various cryptocurrency markets by examining their log returns. To do this, we introduced a new metric called the \textit{trend irreversibility index}, which measures the degree of time-irreversibility based on the analysis of uptrend and downtrend distributions. Uptrends and downtrends in a time series are defined as subsequences that monotonically increase or decrease, respectively. The distributions of the lengths of these trend patterns are used to define the uptrend and downtrend distributions.
We have shown the relationship between the irreversibility index (which is the Kullback-Leibler divergence of the uptrend distribution with respect to the downtrend distribution) and the entropy production rate in a simple Markov chain, namely the random walk on $\mathbb{Z}$. After testing our estimator in classical models of reversible and irreversible autoregressive processes, we estimated the irreversibility of log returns from five cryptocurrency markets: Bitcoin, Ethereum, Ripple, Bitcoin Cash, and Litecoin. Our findings suggest  that all of these cryptocurrency markets exhibit strong signs of irreversibility. Among the cryptocurrencies with the highest market capitalization, BTC, ETH, and XRP appear to follow a global trend towards reversibility over time, whereas LTC displays the opposite behavior, showing an increase in irreversibility over time. Bitcoin Cash, on the other hand, maintains a low degree of irreversibility that remains almost constant over time. It is important to note that the tendencies observed in the irreversibility curves are not absolute; all of the cryptocurrency markets analyzed exhibited periods of irreversibility followed by periods of reversibility. Additionally, at least one cryptocurrency showed a series of increasing irreversibility peaks on the $I_{\mathrm{T}}$ curve as a function of time, and those cryptocurrencies that showed a time evolution towards reversibility did not exhibit strictly monotonic behavior.

We proceeded to assess inefficiency using an index based on Shannon entropy introduced in Ref.~\cite{Brouty2022Statistical} to analyze the relationship between irreversibility and inefficiency. We estimated the index for all the mentioned cryptocurrencies using the sliding window technique to analyze how inefficiency changes over time. Our findings show that inefficiency in the analyzed cryptocurrency markets follows a similar behavior to irreversibility. Some cryptocurrencies seem to evolve towards efficiency in a non-monotonic, and even random, manner, while others do not have a clear tendency. In general, all the inefficiency curves exhibit periods of efficiency followed by periods of inefficiency, which might indicate non-stationarity, similar to what we observed in the time evolution of irreversibility. Although we found that irreversibility and inefficiency curves share some characteristics, we did not find a clear equivalence between the two indices. Some of the cryptocurrencies studied here exhibited periods in which inefficiency and irreversibility seemed to follow each other, although this was not sustained for the entire analyzed time span. Other cryptocurrencies did not show any clear coincidence. To do this comparison systematically, we evaluated the Pearson correlation coefficient, which allowed us to see how much inefficiency and irreversibility are related. The correlation between inefficiency and irreversibility is  not large enough to suggest an equivalence, at least within the statistical extent of this analysis. In any case, our results show that a much deeper study on irreversibility is necessary to understand to what extent irreversibility might provide information on the analyzed market or what are the origins of the observed irreversibility are, for example, by analyzing the price formation through models of limit order books. A more careful analysis of inefficiency, such as in the semi-strong sense or analyzing other inefficiency indices, might be crucial to gain a better understanding of the role of irreversibility in inefficiency for cryptocurrency and other financial markets.

\section*{Acknowledgments}
The authors thanks CONACyT-Mexico by financial support through grant number CF-2019-1327701. JMH was supported by CONACyT via the Doctoral Fellowship number 858831.





\bibliographystyle{elsarticle-num} 
\bibliography{TrendPattBib}





\end{document}